\documentclass[letterpaper, conference]{IEEEtran}
\IEEEoverridecommandlockouts
\usepackage{amsthm}  
\newtheorem{theorem}{Theorem}
\usepackage{mathrsfs} 
\usepackage{pgfplots}  
\pgfplotsset{compat=1.18}
\usepackage{subcaption} 
\usepackage{caption}    
\usepackage{placeins}  
\usepackage{float} 
\usepackage{relsize}
\usepackage{cite}
\usepackage{amsmath,amssymb,amsfonts}
\usepackage{mathtools} 
\usepackage{url}
\usepackage{array}
\usepackage{algorithm}
\usepackage{algpseudocode}
\usepackage{graphicx}
\usepackage{textcomp}
\usepackage{xcolor}
\def\BibTeX{{\rm B\kern-.05em{\sc i\kern-.025em b}\kern-.08em    T\kern-.1667em\lower.7ex\hbox{E}\kern-.125emX}}

\begin{document}


\title{A Feedback Control Framework for Incentivised Suburban Parking Utilisation and Urban Core Traffic Relief}

\author{

\IEEEauthorblockN{Abdul Baseer Satti$^{*}$, James Saunderson$^{*}$, Wynita Griggs$^{*,\dagger}$}
\IEEEauthorblockA{\textit{$^{*}$Department of Electrical and Computer Systems Engineering}\\
\textit{$^{\dagger}$Department of Civil and Environmental Engineering}\\
\textit{Monash University}\\
Clayton, Australia\\
$\{$abdul.satti, james.saunderson, wynita.griggs$\}$@monash.edu}\\

\IEEEauthorblockN{S. M. Nawazish Ali, Nameer Al Khafaf, Saman Ahmadi, Mahdi Jalili}
\IEEEauthorblockA{\textit{School of Engineering}\\
\textit{RMIT University}\\
Melbourne, Australia\\
$\{$syed.ali, nameer.al.khafaf, saman.ahmadi, mahdi.jalili$\}$@rmit.edu.au}\\

\begin{tabular}{c@{\extracolsep{8em}}c} 
\textlarger[0.5]{Jakub Mare\v{c}ek} & \textlarger[0.5]{Robert Shorten} \\
\textit{Artificial Intelligence Center} & \textit{Dyson School of Design Engineering} \\ 
\textit{Czech Technical University in Prague} & \textit{Imperial College London} \\
Prague, Czech Republic & South Kensington, London, United Kingdom\\
jakub.marecek@fel.cvut.cz & r.shorten@imperial.ac.uk
\end{tabular}

\thanks{This work was supported by an Australian Government Research Training Program (RTP) Scholarship; and by the Australian Government Department of Industry, Science and Resources, and the Department of Climate Change, Energy, the Environment and Water, under the International Clean Innovation Researcher Networks (ICIRN) program (Project ID ICIRN000077).}

}

\maketitle

\begin{abstract}
Urban traffic congestion, exacerbated by inefficient parking management and cruising for parking, significantly hampers mobility and sustainability in smart cities. Drivers often face delays searching for parking spaces, influenced by factors such as accessibility, cost, distance, and available services such as charging facilities in the case of electric vehicles. These inefficiencies contribute to increased urban congestion, fuel consumption, and environmental impact. Addressing these challenges, this paper proposes a feedback control incentivisation-based system that aims to better distribute vehicles between city and suburban parking facilities offering park-and-charge/-ride services. Individual driver behaviours are captured via discrete choice models incorporating factors of importance to parking location choice among drivers, such as distance to work, public transport connectivity, charging infrastructure availability, and amount of incentive offered; and are regulated through principles of ergodic control theory. The proposed framework is applied to an electric vehicle park-and-charge/-ride problem, and demonstrates how predictable long-term behaviour of the system can be guaranteed.
\end{abstract}


\section{Introduction}

Rapid growth in many urban centres around the world, coupled with ineffective policies \cite{b-intro-5}, has led to significant increases in the use of privately-owned cars \cite{b-intro-6}, making parking management a critical challenge in modern cities. These challenges are further exasperated by the rise of electric vehicles, which require charging facilities in addition to space to park. Intelligent Transportation Systems (ITS) are able to offer promising solutions by leveraging increasingly sophisticated models, hardware and computing technologies, real-time data, and dynamic controls to optimise parking and traffic management \cite{b-intro-1,b-intro-2,b-intro-3}. However, aligning individual driver decisions with the broader objectives of local authorities introduces complexities that continue to demand advanced strategies and innovative approaches.



Iterated Function Systems (IFS) and ergodic control theory can provide powerful approaches to shaping collective agent behaviour by offering robust frameworks for modelling stochastic dynamics, capturing probabilistic behaviours of multi-agent systems; see, for example, \cite{b4,b-intro-7,b-intro-11,b-intro-10,b-intro-8,b-intro-9}. While IFS have demonstrated effectiveness in a range of applications, from image processing \cite{b-intro-12}, to power systems operations \cite{b-intro-11}, their potential use for regulating parking as a resource in urban environments remains unexplored. In this paper, we therefore utilise concepts from IFS and ergodic control theory to present a novel feedback system that seeks to regulate, through incentivisation, parking location choice among drivers, where drivers have the capacity to choose to park in a city's central business district, close to a place of work, or in one of multiple suburbs that offer park-and-charge/-ride facilities.


In our framework, the incentive amounts offered to drivers are the outputs of decentralised feedback controllers, each associated with a parking location, and each potentially managed by a different local authority; and we demonstrate that, under some straightforward conditions, our system can guarantee predictable system behaviour, on the average, in the long-term, even in the face of the stochastic nature of human decision-making. Our framework employs discrete choice modelling (specifically, multinomial logit models) to model driver behaviour in regard to parking location choice. In Section \ref{s2} of the paper, we elaborate further on the problem scenario. In Section \ref{s3}, we present the novel system framework, along with the human decision-making model utilised, and the required mathematical results. A demonstrative park-and-charge/-ride use case scenario is then presented in Section \ref{s4}. Finally, we provide avenues for future work in Section \ref{s5}.

\subsection{Notation}
Generally speaking, superscripts will denote drivers, and subscripts will denote parking locations. $T$ denotes the matrix transpose.

\section{Problem Description}\label{s2}

We begin by describing a situation often encountered on a weekday morning in urban areas. Consider the scenario depicted in Fig. \ref{fig:scenario}. Suppose that there are $N$ privately-owned vehicles inhabited by occupants needing to travel into a City's central business district (CBD) during peak hour to get to work. By default, these occupants nominally prefer to drive their cars all of the way into the City, and park there, because doing so is an attractive option in terms of travel comfort, convenience and habit. 

\begin{figure}[b]
    \centering
\includegraphics[width=\columnwidth]{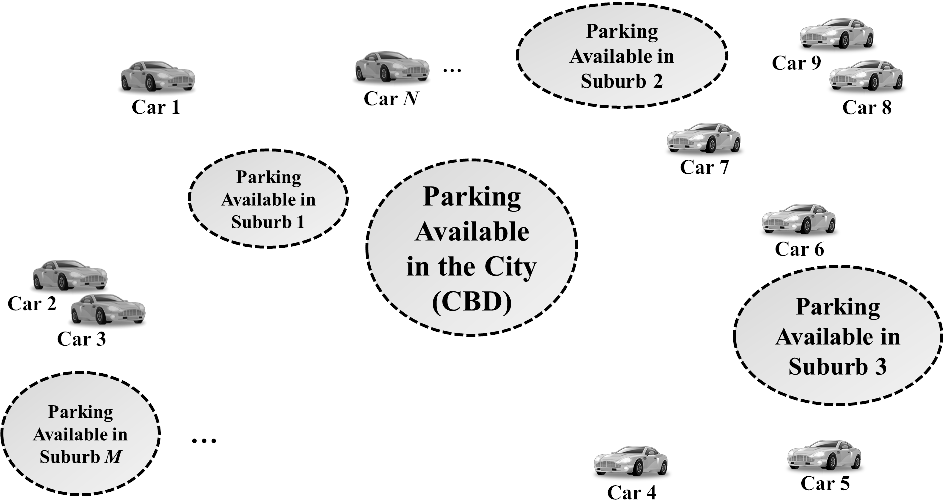}
    \caption{Problem scenario. (Sub-images of the car obtained from Openclipart \cite{b1}.)}
    \label{fig:scenario}
\end{figure}

Conversely, suppose local government authorities desire to remove vehicle congestion from the City by decreasing demand for City parking spaces; and thus would like to incentivise some of the privately-owned cars to instead park in one of the $M$ Suburbs, such that the occupants of the privately-owned vehicles then use alternative means of transport to complete their journeys into the City (e.g., public transport). One could also suppose that suburban government authorities have at least some level of autonomy over any potential incentivisation schemes implemented in their own geographical areas.

Our aim is to introduce a feedback control incentivisation-based system that achieves these goals; and moreover, we will demonstrate that, under some straightforward conditions, one can use ideas from ergodic control theory to guarantee predicability of the system, in the long-term, regardless of its initial conditions.


\section{Proposed Framework}\label{s3}


Our proposed framework takes the form of the feedback control system illustrated in Fig. \ref{fig:interconnection}. Recall that $N$ and $M$ are positive integers, and let $M<N$. We elaborate on each of the components of the block diagram, as follows:
\begin{itemize}
    \item let $r:=\left(r_1 \cdots r_M \right)^T$ denote the column vector of reference inputs, where each $r_j$, for $j=1,\dots,M$, is the number of cars that are desired by local authorities to park in location $j$;
    \item let $\mathcal{C}:=\textbf{diag}\left(\mathcal{C}_1,\dots,\mathcal{C}_M \right)$ denote a matrix comprised of linear shift-invariant, single-input, single-output controllers, $\mathcal{C}_j$, down its diagonal, and entries of 0 elsewhere, where each $\mathcal{C}_j$ takes as input an error, $e_j$, and produces an output, $\pi_j$, which represents an incentive amount to be offered to each driver to entice them to park in location $j$;
    \item suppose \begin{equation}
        \mathcal{F}:=\begin{pmatrix}
        \mathcal{F}_1 & 0 & \cdots & 0 & 0 \\
        0 & \mathcal{F}_2 & \cdots & 0 & 0 \\
        \vdots & \vdots & \cdots & \vdots & \vdots \\
        0 & 0 & \cdots & \mathcal{F}_M & 0
    \end{pmatrix}\label{e-filter}
    \end{equation} denotes a matrix comprised of linear shift-invariant, single-input, single-output filters (e.g., delay operators, moving average filters), $\mathcal{F}_j$, and entries of 0, as indicated by \eqref{e-filter};
    \item let $\hat{y}:=\left(\hat{y}_1 \cdots \hat{y}_M \right)^T$ denote the output of $\mathcal{F}$, $e:=\left(e_1 \cdots e_M \right)^T = r-\hat{y}$ and $\pi:=\left(\pi_1 \cdots \pi_M \right)^T$;
    \item for each $i=1,\dots,N$, let $\mathcal{P}^i$ denote the $i$th driver, who takes as input the column vector of all offered incentive amounts, $\pi$, and outputs a column vector, $y^i$, from the set
    \begin{equation}y^i\in\left\{\begin{pmatrix}
        1_1\\0\\\vdots\\0 \end{pmatrix}, \begin{pmatrix}
        0\\1_2\\\vdots\\0 \end{pmatrix}, \dots, \begin{pmatrix}
        0\\0\\\vdots\\1_{M+1} \end{pmatrix}\right\},\end{equation}
   where the subscript of the nonzero entry in $y^i$ corresponds to the parking location driver $i$ has decided upon at time step $k$, with parking location $M+1$ denoting the City;
   \item and, finally, note that $y=\sum^{N}_{i=1} y^i$.
\end{itemize}

\begin{figure}[t]
    \centering
\includegraphics[width=\columnwidth]{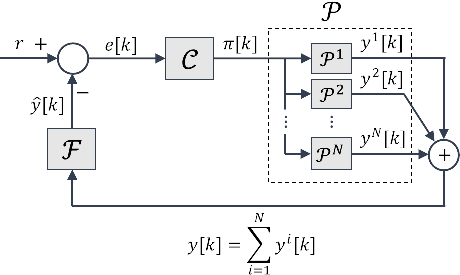}
    \caption{Block diagram of the feedback control incentivisation-based system.}
    \label{fig:interconnection}
\end{figure}

\subsection{Driver decision-making model}

A discrete choice model is used to describe the manner by which a driver, $\mathcal{P}^i$, produces an output, $y^i$, noting that, in our scenario, the set of parking location choices (i.e., the finite set consisting of the City, and each of the $M$ Suburbs) is mutually exclusive and collectively exhaustive. Specifically, let the probability, $p^i_{\text{Suburb}_j}(\pi_1,\dots,\pi_M)$, of a driver $i$ deciding to park in Suburb $j$ (and thus consume that resource), for each $j=1,\dots,M$, be given by
\begin{multline}
    p^i_{\text{Suburb}_j}(\pi_1,\dots,\pi_M) = \\
    \frac{e^{\text{Utility}^i_j(\pi_j)}}{e^{\text{Utility}^i_1(\pi_1)} + \dots + e^{\text{Utility}^i_M(\pi_M)} + e^{\text{Utility}^i_{\text{City}}}};
\end{multline}
and the probability of driver $i$ deciding to park in the City, $p^i_{\text{City}}(\pi_1,\dots,\pi_M)$, be given by
\begin{multline}
    p^i_{\text{City}}(\pi_1,\dots,\pi_M) = \\
    \frac{e^{\text{Utility}^i_{\text{City}}}}{e^{\text{Utility}^i_1(\pi_1)} + \dots + e^{\text{Utility}^i_M(\pi_M)} + e^{\text{Utility}^i_{\text{City}}}};
\end{multline}
where each $\pi_j[k]$ is the incentive amount being offered by a local authority as enticement to drivers to park in Suburb $j$, at discrete time step $k$. Note that Utility$^i_{j}(\pi_j)$, as well as Utility$^i_{\text{City}}$, are functions representing the perceived attractiveness of the different parking choice locations to driver $i$. Throughout this paper, we will assume utility functions of the form
\begin{equation}
    \text{Utility}^i_{j}(\pi_j) = \gamma^i_{j0} \times \pi_j + \sum_{k=1}^{Q\in\mathbb{N}}\left(\gamma^i_{jk} \times X^i_{jk}\right),
\end{equation}
where the $X^i_{jk}$ are attributes (e.g., total travel time to work, parking location fees, public transportation fees, cost of petrol, cost to charge electric vehicle, income level of the driver, frequency of public transport services, number of electric vehicle chargers available) associated with each of the alternative parking locations and/or individual drivers; and $\gamma^i_{j0}$, and $\gamma^i_{jk}$, for $k=1,\dots,Q$, are weights relating to how much significance driver $i$ places on the corresponding $k$th attribute. Similarly, let
\begin{equation}
    \text{Utility}^i_{\text{City}} = bias^i + \sum_{k=1}^{Q\in\mathbb{N}}\left(\gamma^i_{\text{City}k} \times X^i_{\text{City}k}\right),
\end{equation}
where $bias^i$ is a weight, or bias constant, signifying an inherit preference for driving one's privately-owned car all of the way into the City and parking there.

For each driver $i$, also note that
\begin{equation}
    p^i_{\text{City}}(\pi_1,\dots,\pi_M) + \sum^{M}_{j=1}p^i_{\text{Suburb}_j}(\pi_1,\dots,\pi_M) = 1.
\end{equation}

As an example, for illustrative purposes, the probabilities $p_{\text{Suburb}_1}(\pi_1,\pi_2)$, $p_{\text{Suburb}_2}(\pi_1,\pi_2)$ and $p_{\text{City}}(\pi_1,\pi_2)$, as well as the sum of all three of these probabilities, have been plotted in Fig. \ref{fig:probabilities}, for $\gamma_{10} = \gamma_{20} = 10$, $\sum\left(\gamma_{1\bullet} \times X_{1\bullet}\right) = -62.28$, $\sum\left(\gamma_{2\bullet} \times X_{2\bullet}\right) = -66$, $bias = 35$ and $\sum\left(\gamma_{\text{City}\bullet} \times X_{\text{City}\bullet}\right) = -53.12$. From Fig. \ref{fig:probabilities}, it can be seen (for instance) that, for a fixed $\pi_2$, if enough incentive $\pi_1$ is offered to a driver to park in Suburb 1, then they are more likely to do so; and similarly, for a fixed $\pi_1$, if enough incentive $\pi_2$ is offered to a driver to park in Suburb 2, then they are more likely to do so. Conversely, and as desired, offering enough of either incentive $\pi_1$ and/or $\pi_2$ increases the likelihood of a driver being enticed to not park in the City.

\begin{figure}[b]
    \centering
\includegraphics[width=\columnwidth]{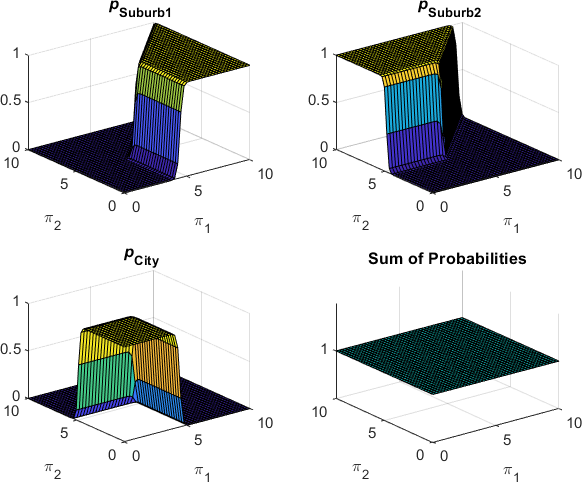}
    \caption{Probabilities versus incentives.}
    \label{fig:probabilities}
\end{figure}

It is also important to note that the amount of parking space in the Suburbs and the City is a finite resource. Therefore, as indicated in our proposed framework above, we employ controllers in a feedback-loop as a means to attempt to regulate the numbers of drivers choosing which parking locations to go to, in the face of their probabilistic decision-making, where the controller outputs are the incentive amounts. Due to the stochastic nature of human decision-making, we therefore use results from ergodic control theory and IFS to achieve this regulation, through guaranteeing the predictability of our system in the long-run. In particular, we invoke Theorem 12 of \cite{b4}, as seen in the next subsection.

\subsection{Mathematical Result}

Utilising Theorem 12 of \cite{b4}, we derive the following result.

\begin{theorem}\label{thm1}
Given $M$ stable, linear shift-invariant, single-input, single-output controllers, $\mathcal{C}_j$, $j=1,\dots,M$, where the dynamics of each controller are described by
\begin{align}
x_{\mathcal{C}_j}[k+1] & = A_{\mathcal{C}_j} x_{\mathcal{C}_j}[k] + B_{\mathcal{C}_j} e_j[k],\\
\pi_j[k] & = C_{\mathcal{C}_j} x_{\mathcal{C}_j}[k] + D_{\mathcal{C}_j} e_j[k],
\end{align}
with $x_{\mathcal{C}_j}$ denoting the internal state of the controller; and given $M$ stable, linear shift-invariant, single-input, single-output filters, $\mathcal{F}_j$, comprising $\mathcal{F}$ as described at the beginning of Section \ref{s3}, where the dynamics of $\mathcal{F}$ are expressed by
\begin{align}
x_{\mathcal{F}}[k+1] & = A_{\mathcal{F}} x_{\mathcal{F}}[k] + B_{\mathcal{F}} y[k],\\
\hat{y}[k] & = C_{\mathcal{F}} x_{\mathcal{F}}[k],
\end{align}
with $x_{\mathcal{F}}$ denoting the internal state of $\mathcal{F}$; then consider the feedback system depicted in Fig. \ref{fig:interconnection}. Suppose that each driver $i$, where $i=1,\dots,N$, has dynamics governed by
\begin{align}
x^i[k+1] & = b^i,\\
y^i[k] & = x^i[k],
\end{align}
where $x^i$ denotes the internal state of the $i$th driver, and $b^i$ is a random variable, selected at each time step from the set
\begin{multline}
    \left\{b^{i}_{\text{Suburb}_1} := \begin{pmatrix} 1_1\\0\\\vdots\\0 \end{pmatrix}, b^{i}_{\text{Suburb}_2} := \begin{pmatrix} 0\\1_2\\\vdots\\0 \end{pmatrix}, \dots, \right. \\ \left. b^{i}_{\text{City}} := \begin{pmatrix} 0\\0\\\vdots\\1_{M+1} \end{pmatrix}\right\}\label{set}
\end{multline}
according to $\mathbb{P}(b^i=b^{i}_{r})=p^{i}_r(\pi_1,\dots,\pi_M)$, for all $r\in\{\text{Suburb}_1,\dots,\text{Suburb}_M,\text{City}\}$. Then, the feedback loop converges in distribution to a unique invariant measure.
\end{theorem}
\begin{proof}
    Theorem \ref{thm1} is a modification of \cite[Theorem 12]{b4}. First, note that, if each linear shift-invariant, single-input, single-output controller, $\mathcal{C}_j$, is stable, then $\mathcal{C}=\textbf{diag}\left(\mathcal{C}_1,\dots,\mathcal{C}_M \right)$ is stable. This can be seen by expressing $\mathcal{C}$ as
    \begin{align}
        \mathcal{C} &= \textbf{diag}\left(C_{\mathcal{C}_1}(zI-A_{\mathcal{C}_1})^{-1}B_{\mathcal{C}_1} + D_{\mathcal{C}_1}, \dots, \right.\notag\\
        &\hspace{12.2mm} \left. C_{\mathcal{C}_M}(zI-A_{\mathcal{C}_M})^{-1}B_{\mathcal{C}_M} + D_{\mathcal{C}_M} \right)\\
        &= \textbf{diag}\left(C_{\mathcal{C}_1},\dots,C_{\mathcal{C}_M} \right)\notag\\
        &\hspace{12.2mm} \times \left(\textbf{diag}\left(zI-A_{\mathcal{C}_1},\dots,zI-A_{\mathcal{C}_M} \right)\right)^{-1}\notag\\
        &\hspace{12.2mm} \times \textbf{diag}\left(B_{\mathcal{C}_1},\dots,B_{\mathcal{C}_M} \right)\notag\\
        &\hspace{12.2mm} + \textbf{diag}\left(D_{\mathcal{C}_1},\dots,D_{\mathcal{C}_M} \right)\\
        &= \textbf{diag}\left(C_{\mathcal{C}_1},\dots,C_{\mathcal{C}_M} \right)\notag\\
        &\hspace{12.2mm} \times \left(zI-\textbf{diag}\left(A_{\mathcal{C}_1},\dots,A_{\mathcal{C}_M} \right)\right)^{-1}\notag\\
        &\hspace{12.2mm} \times \textbf{diag}\left(B_{\mathcal{C}_1},\dots,B_{\mathcal{C}_M} \right)\notag\\
        &\hspace{12.2mm} + \textbf{diag}\left(D_{\mathcal{C}_1},\dots,D_{\mathcal{C}_M} \right),
    \end{align}
    where a matrix inversion formula \cite[Section 2.3]{b3} was repeatedly employed to derive the second equality; and subsequently observing that, since the spectral radius of each $A_{\mathcal{C}_j}$ is strictly less than one, then the spectral radius of $\textbf{diag}\left(A_{\mathcal{C}_1},\dots,A_{\mathcal{C}_M} \right)$ is strictly less than one. Similarly, it can be shown that, if each linear shift-invariant, single-input, single-output filter, $\mathcal{F}_j$, is stable, then $\mathcal{F}$ (as defined at the beginning of Section \ref{s3}) is stable. The remainder of the proof follows in the manner of \cite[Theorem 12]{b4}'s proof.
\end{proof}


\section{An Example Use Case Scenario: Park-and-Charge/-Ride}\label{s4}

The following use case scenario serves to demonstrate the efficacy of our system. Let us consider a small urban region with suitable park-and-charge/-ride locations situated in Suburbs 1 and 2 and the City. In terms of public transport, the region is serviced by buses.

Let $N=100$ be the total number of privately-owned vehicles inhabited by occupants wanting to travel into the City during the morning's peak hour to get to work. Furthermore, suppose that all of the privately-owned vehicles can be categorised into one of two classes: (Class 1) those that are electric vehicles; and (Class 2) those that are not, thus being vehicles with internal combustion engines. Suppose that all drivers with vehicles in Class 1 (i.e., Population 1) have the following utility functions associated with them
\begin{gather*}
    \text{Utility}^{i_{\text{Class}1}}_{1}(\pi_1) = \gamma^{i_{\text{Class}1}}_{10} \times \pi_1 + \sum_{k=1}^{6}\left(\gamma^{i_{\text{Class}1}}_{1k} \times X^{i_{\text{Class}1}}_{1k}\right),\\
    \text{Utility}^{i_{\text{Class}1}}_{2}(\pi_2) = \gamma^{i_{\text{Class}1}}_{20} \times \pi_2 + \sum_{k=1}^{6}\left(\gamma^{i_{\text{Class}1}}_{2k} \times X^{i_{\text{Class}1}}_{2k}\right),\\
    \text{Utility}^{i_{\text{Class}1}}_{\text{City}} = bias^{i_{\text{Class}1}} + \sum_{k=1}^{6}\left(\gamma^{i_{\text{Class}1}}_{\text{City}k} \times X^{i_{\text{Class}1}}_{\text{City}k}\right);
\end{gather*}
and all drivers with vehicles in Class 2 (i.e., Population 2) have the following utility functions associated with them
\begin{gather*}
    \text{Utility}^{i_{\text{Class}2}}_{1}(\pi_1) = \gamma^{i_{\text{Class}2}}_{10} \times \pi_1 + \sum_{k=1}^{6}\left(\gamma^{i_{\text{Class}2}}_{1k} \times X^{i_{\text{Class}2}}_{1k}\right),\\
    \text{Utility}^{i_{\text{Class}2}}_{2}(\pi_2) = \gamma^{i_{\text{Class}2}}_{20} \times \pi_2 + \sum_{k=1}^{6}\left(\gamma^{i_{\text{Class}2}}_{2k} \times X^{i_{\text{Class}2}}_{2k}\right),\\
    \text{Utility}^{i_{\text{Class}2}}_{\text{City}} = bias^{i_{\text{Class}2}} + \sum_{k=1}^{6}\left(\gamma^{i_{\text{Class}2}}_{\text{City}k} \times X^{i_{\text{Class}2}}_{\text{City}k}\right);
\end{gather*}
where $i_{\text{Class}1}=1,\dots,20$, $i_{\text{Class}2}=1,\dots,80$, $bias^{i_{\text{Class}1}} = bias^{i_{\text{Class}2}} = 35$ and $\gamma^{i_{\text{Class}1}}_{10} = \gamma^{i_{\text{Class}1}}_{20} = \gamma^{i_{\text{Class}2}}_{10} = \gamma^{i_{\text{Class}2}}_{20} = 10$.

Moreover, suppose $\sum_{k=1}^{6}\left(\gamma^{i_{\text{Class}1}}_{1k} \times X^{i_{\text{Class}1}}_{1k}\right)$ $=$ $-62.28$, $\sum_{k=1}^{6}\left(\gamma^{i_{\text{Class}1}}_{2k} \times X^{i_{\text{Class}1}}_{2k}\right)$ $=$ $-66$, $\sum_{k=1}^{6}\big(\gamma^{i_{\text{Class}1}}_{\text{City}k} \times X^{i_{\text{Class}1}}_{\text{City}k}\big)$ $=$ $-53.12$, $\sum_{k=1}^{6}\left(\gamma^{i_{\text{Class}2}}_{1k} \times X^{i_{\text{Class}2}}_{1k}\right)$ $=$ $-51.5$, $\sum_{k=1}^{6}\left(\gamma^{i_{\text{Class}2}}_{2k} \times X^{i_{\text{Class}2}}_{2k}\right)$ $=$ $-61$ and $\sum_{k=1}^{6}\big(\gamma^{i_{\text{Class}2}}_{\text{City}k} \times X^{i_{\text{Class}2}}_{\text{City}k}\big)$ $=$ $-35$, where $X^{\bullet}_{\bullet 1}$ denotes total travel time to work, $X^{\bullet}_{\bullet 2}$ denotes the parking fees payable at the location, $X^{\bullet}_{\bullet 3}$ denotes the additional fees payable to charge an electric vehicle at the parking location, $X^{\bullet}_{\bullet 4}$ denotes the ticket price of taking a bus from a suburban parking location to the City, $X^{\bullet}_{\bullet 5}$ denotes the frequency of bus services from a suburban parking location to the City, and $X^{\bullet}_{\bullet 6}$ denotes the number of electric vehicle chargers at the parking location.

Let $r_1=25$ and $r_2=35$ be the number of vehicles that local authorities would like to see accommodated in Suburbs 1 and 2, respectively, as part of their park-and-charge/-ride programs. Each local authority has autonomy over the incentive amounts that they offer to drivers, in that each local authority $j$ implements a controller $\mathcal{C}_j$ of the form
\begin{equation}
    \pi_j[k] = \beta_j  \pi_j[k-1] + \kappa_j(e_j[k] - \alpha_je_j[k-1]),
\end{equation}
where $\alpha_1 = -0.01$, $\alpha_2 = -0.01$, $\beta_1 = 0.9$, $\beta_2 = 0.99$, $\kappa_1 = 0.15$ and $\kappa_2 = 0.2$, for $j=1,2$.

For the experiment, 1,000 simulations were run in total, and each simulation ran for 1,000 time steps. For simplicity, the controllers updated at every time step; and the filters $\mathcal{F}_{1}$ and $\mathcal{F}_{2}$ each invoked a pure delay of one time step. At the beginning of each simulation run, the six initial values for $p^{i_{\text{Class}1}}_{\text{Suburb}_1}(\pi_1,\pi_2)$, $p^{i_{\text{Class}1}}_{\text{Suburb}_2}(\pi_1,\pi_2)$, $p^{i_{\text{Class}1}}_\text{City}(\pi_1,\pi_2)$, $p^{i_{\text{Class}2}}_{\text{Suburb}_1}(\pi_1,\pi_2)$, $p^{i_{\text{Class}2}}_{\text{Suburb}_2}(\pi_1,\pi_2)$ and $p^{i_{\text{Class}2}}_\text{City}(\pi_1,\pi_2)$ were randomly generated to demonstrate different initial conditions.


Figs. \ref{fig:results_1} and \ref{fig:results_2} illustrate the average numbers of drivers (from each of the populations) deciding upon Suburbs 1 and 2, respectively, as their parking location of choice, over time, from the 1,000 simulation runs. Figs. \ref{fig:results_3} and \ref{fig:results_4} depict the mean outputs from controllers $\mathcal{C}_{1}$ and $\mathcal{C}_{2}$, respectively, over time. Figs. \ref{fig:results_5} and \ref{fig:results_6} illustrate the mean inputs to controllers $\mathcal{C}_{1}$ and $\mathcal{C}_{2}$, respectively, over time. Convergence of the means in all cases, and thus predictability of the system, is evident. The controllers did a reasonable job at achieving small mean steady state errors (see Figs. \ref{fig:results_5} and \ref{fig:results_6}), noting that the use of lag controllers was given preference over the use of PI controllers given the latter's potential to introduce instabilities.

\begin{figure}[t]
    \centering
\includegraphics[width=\columnwidth]{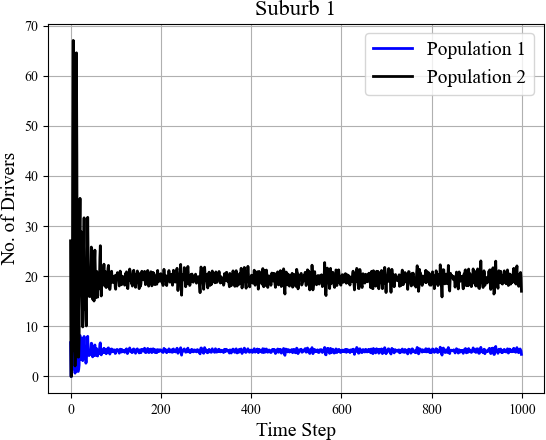}
    \caption{Evolution of the mean number of drivers, from each population, choosing Suburb 1. The population of drivers with vehicles in Class 1 are represented by the blue line, while the black line represents the population of divers with vehicles in Class 2.}
    \label{fig:results_1}
\end{figure}

\begin{figure}[t]
    \centering
\includegraphics[width=\columnwidth]{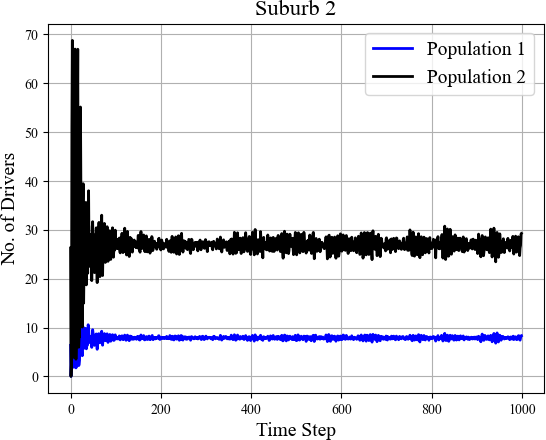}
    \caption{Evolution of the mean number of drivers, from each population, choosing Suburb 2. The population of drivers with vehicles in Class 1 are represented by the blue line, while the black line represents the population of divers with vehicles in Class 2.}
    \label{fig:results_2}
\end{figure}

\begin{figure}[t]
    \centering
\includegraphics[width=\columnwidth]{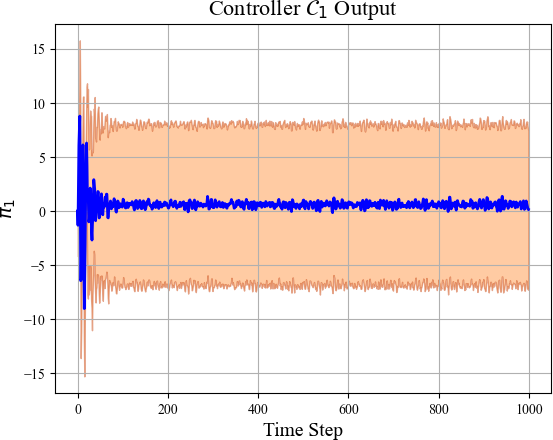}
    \caption{Output, $\pi_1$, of controller $\mathcal{C}_{1}$ over time. The mean output is represented by the blue line, while the orange shaded area indicates one standard deviation from the mean.}
    \label{fig:results_3}
\end{figure}

\begin{figure}[t]
    \centering
\includegraphics[width=\columnwidth]{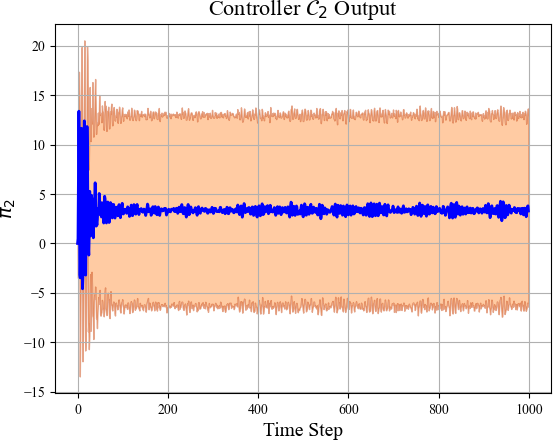}
    \caption{Output, $\pi_2$, of controller $\mathcal{C}_{2}$ over time. The mean output is represented by the blue line, while the orange shaded area indicates one standard deviation from the mean.}
    \label{fig:results_4}
\end{figure}

\begin{figure}[t]
    \centering
\includegraphics[width=\columnwidth]{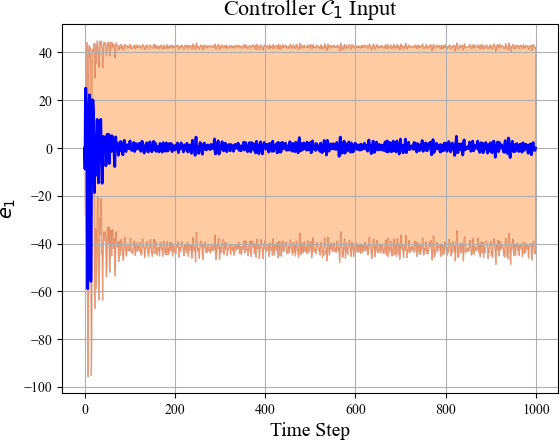}
    \caption{Input, $e_1$, to controller $\mathcal{C}_{1}$ over time. The mean input is represented by the blue line, while the orange shaded area indicates one standard deviation from the mean.}
    \label{fig:results_5}
\end{figure}

\begin{figure}[t]
    \centering
\includegraphics[width=\columnwidth]{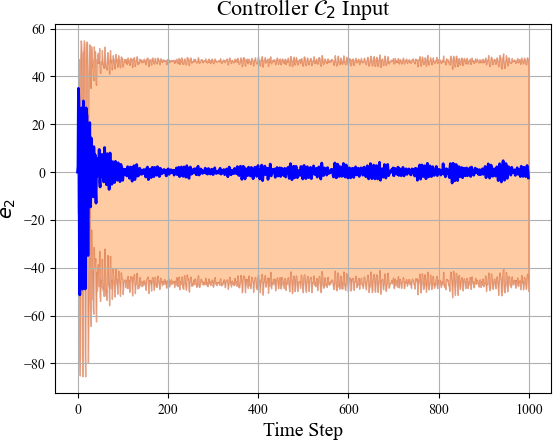}
    \caption{Input, $e_2$, to controller $\mathcal{C}_{2}$ over time. The mean input is represented by the blue line, while the orange shaded area indicates one standard deviation from the mean.}
    \label{fig:results_6}
\end{figure}

\section{Conclusions and Future Work}\label{s5}

In this paper, using tools from ergodic control theory and discrete choice modelling, a novel feedback system that regulates parking location choice among drivers, through incentivisation, was introduced. The proposed framework was applied to a park-and-charge/-ride scenario, for which it was shown that predictable long-term behaviour of the system was guaranteed. Following on from this study, there exist many potential and interesting avenues for future work to be explored. Some of these avenues consist of increasing the complexity and variety of the use case scenarios, including integration of the park-and-charge/-ride scenario with power grid features; and comparing and contrasting our proposed framework to existing incentivisation techniques.






\begin{thebibliography}{00}

\bibitem{b-intro-5} E. J. Taylor and R. van Bemmel-Misrachi, The elephant in the scheme: planning for and around car parking in Melbourne, 1929–2016, \emph{Land Use Policy}, vol. 60, pp. 287–297, 2017.

\bibitem{b-intro-6} M. Chwasta, \emph{Australia's biggest city has a car problem. What should Melbourne do to fix it?}, Article, ABC News, 15 April 2024. (Last accessed: 30 April 2025.) [Online]. Available: https://www.abc.net.au/news/2024-04-15/vehicle-car-ownership-registrations-traffic-congestion-melbourne/103685800

\bibitem{b-intro-1} I. Bogoslavskyi, L. Spinello, W. Burgard and C. Stachniss, Where to park? Minimizing the expected time to find a parking space, in \emph{Proceedings of the 2015 IEEE International Conference on Robotics and Automation}, Seattle, WA, USA, 2015, pp. 2147–2152.

\bibitem{b-intro-2} A. Schlote, C. King, E. Crisostomi and R. Shorten, Delay-tolerant stochastic algorithms for parking space assignment, \emph{IEEE Transactions on Intelligent Transportation Systems}, vol. 15, no. 5, pp. 1922-1935, 2014.

\bibitem{b-intro-3} H. M. Abdelghaffar, S. F. A. Batista, A. Rehman, J. Cao, M. Men\'{e}ndez and S. E. Jabari, Comparison of probabilistic cruising-for-parking time estimation models, \emph{Transportation Research Part A: Policy and Practice}, vol. 184, 2024.



\bibitem{b4} A. R. Fioravanti, J. Mare\v{c}ek, R. N. Shorten, M. Souza and F. R. Wirth, On the ergodic control of ensembles, \emph{Automatica}, vol. 108, 2019.

\bibitem{b-intro-7} R. Ghosh, J. Mare\v{c}ek, W. M. Griggs, M. Souza and R. N. Shorten, Predictability and fairness in social sensing, \emph{IEEE Internet of Things Journal}, vol. 9, no. 1, pp. 37-54, 2022

\bibitem{b-intro-11} J. Mare\v{c}ek, M. Roubalik, R. Ghosh, R. N. Shorten and F. R. Wirth, Predictability and fairness in load aggregation and operations of virtual power plants, \emph{Automatica}, vol. 147, 2023.

\bibitem{b-intro-10} V. Kungurtsev, J. Mare\v{c}ek, R. Ghosh  and R. Shorten, On the ergodic control of ensembles in the presence of non-linear filters, \emph{Automatica}, vol. 152, 2023.

\bibitem{b-intro-8} M. Woodburn, W. M. Griggs, J. Mare\v{c}ek and R. N. Shorten, Herd Routes: a feedback control-based preventative system for improving female pedestrian safety on city streets, \emph{International Journal of Control}, vol. 98, no. 5, pp. 1032-1045, 2025.

\bibitem{b-intro-9} W. M. Griggs, R. Ghosh, J. Mare\v{c}ek and R. N. Shorten, Unique ergodicity in feedback interconnections of ensembles of agents, \emph{International Journal of Control}, to appear. Available: https://doi.org/10.1080/00207179.2025.2469281

\bibitem{b-intro-12} S. Abdulla and K. Mahipal Reddy, Optimizing the neural network and iterated function system parameters for fractal approximation using a modified evolutionary algorithm, \emph{Scientific Reports}, vol. 15, 2025. 

\bibitem{b1} OpenClipart.org. In: https://openclipart.org/detail/202922 (Accessed: 21 April 2025).

\bibitem{b3} K. Zhou with J. C. Doyle and K. Glover, \emph{Robust and Optimal Control}, Prentice Hall, Upper Saddle River, NJ; 1996.


\end{thebibliography}
\end{document}